\documentclass[12pt, draftclsnofoot, onecolumn]{IEEEtran}

\usepackage{amsmath}
\usepackage{amssymb}
\usepackage{amsthm}
\usepackage{cite}

\usepackage{graphicx}
\usepackage{enumerate}
\usepackage{setspace}
\usepackage{subfigure}
\usepackage{color}
\usepackage{multirow}
\usepackage{hhline}
\usepackage{algpseudocode}
\usepackage{algorithm}

\newtheorem{theorem}{{Theorem}}

\title{On the achievable region for interference networks with point-to-point codes}

\author{Jung Hyun Bae, Jungwon Lee, Inyup Kang\\
Mobile Solutions Lab\\
Samsung US R$\&$D Center\\
San Diego, CA, USA\\
Email: jungbae@umich.edu, jungwon@alumni.stanford.edu, inyup.kang@samsung.com}

\begin{document}
\maketitle
\begin{abstract}
This paper studies evaluation of the capacity region for interference networks with point-to-point (p2p) capacity-achieving codes. Such capacity region has recently been 
characterized as union of several sub-regions each of which has distinctive operational characteristics. Detailed evaluation of this region, therefore, can be accomplished in a very
simple manner by acknowledging such characteristics, which, in turn, provides an insight for a simple implementation scenario. Completely generalized message 
assignment which is also practically relevant is considered in this paper, and it is shown to provide strictly larger achievable rates than what traditional message assignment does 
when a receiver with joint decoding capability is used.
\end{abstract}

\section{Introduction}
Consider the following simple channel model.
\begin{subequations}
\begin{eqnarray}
y_{1,t}&=&h_{11}x_{1,t}+h_{12}x_{2,t}+z_{1,t}\\
y_{2,t}&=&h_{21}x_{1,t}+h_{22}x_{2,t}+z_{2,t},
\end{eqnarray}
\end{subequations}
where at time $t=1,2,3,...$, $x_{i,t}$ is signal from transmitter $i$, $y_{i,t}$ is channel output at receiver $i$, $z_{i,t}$ is background noise at receiver $i$, $h_{ij}$ is channel from transmitter $j$ to receiver $i$, which is a complex number. If we assume that $z_{i,t} \sim \mathcal{CN}(0,1)$, and that $x_{i,t}$ is dedicated to receiver $i$, then this channel model corresponds to single-input single-output (SISO) Gaussian interference channel (GIC). We also impose power restriction $P$ on each transmitter.
SISO GIC is one of the simplest network channel models and still possesses essence of cell-edge scenario of the standard cellular system. Regarding fundamental understanding in terms of the maximum achievable rate, i.e., capacity, even this simple scenario fails to give a conclusive answer so far. What we do know, however, is using point-to-point (p2p) capacity achieving codes (p2p codes, from now on) would not achieve the capacity in this scenario. Nevertheless, we are still interested in discussing achievable rates with p2p codes in this paper. \\
\indent
There are two main reasons why we consider p2p codes. One is manageable analysis and the other is easy implementation. First to note is that we have a conclusive answer on the capacity region with p2p codes in SISO Gaussian case as seen in~\cite{BaGaTs11}. According to recent discovery~\cite{BaGaKi12}, however, we also have such conclusive answer for more sophisticated Han-Kobayashi scheme~\cite{HaKo81} which is also the best known scheme so far. The problem with the latter is that it is generally regarded as not computable. We may decide to work on computable sub-region of it as done in~\cite{Sa04, ShCh07, EtTsWa08}, but what we hope to have in this paper is `easily' computable region which provides insights for `simple' implementation which is an aspect we would like to emphasize to have. \\
\indent
Regarding implementation, it is not difficult to see easiness of it with p2p codes. Basically, they are there already, and there is nothing to change encoding-wise and decoding-wise. A very simple version of Han-Kobayashi scheme as in~\cite{EtTsWa08} would be implementable as well, but its performance in practical transmit power range turns out not to be quite beneficial~\cite{BaLeKa11}. In order to obtain significant gain over p2p codes, one might need to resort to more sophisticated version of Han-Kobayashi scheme which involves coded time sharing or modified frequency/time division multiplexing (FDM/TDM) as discussed in~\cite{Sa04, ShCh07}, but such algorithm is not quite implementation-friendly as hinted by aforementioned computability issue.    \\
\indent
When dealing with p2p codes in this paper, we do not restrict ourselves to traditional interference channel, which means $x_{i,t}$ is not necessarily dedicated to the receiver $i$. In 2 user case, let $(a_1,a_2)$ be the tuple which describes an intended receiver of each transmitter in which $a_1$ represents the intended receiver of the transmitter 1, and $a_2$ represents the intended receiver of the transmitter 2. $a_1$ and $a_2$ can have values $0,1,2$ where $0$ implies that the corresponding transmitter does not transmit. For example, $(0,1)$ means that the transmitter 1 does not transmit, and the transmitter 2 transmits to the receiver 1. In this case, the receiver 2 does nothing. In this way, there are total of 8 possible combinations in 2 user case, which are $(0,1),(0,2),(1,0),(1,1),(1,2),(2,0),(2,1),(2,2)$. One important thing to note is that this setup does not necessarily imply extensive use of backhaul communication in the cellular system. In fact, this kind of arbitrary message assignment can be done by higher layer protocol without any use of backhaul communication, and such relevance of this setup in simple implementation is another thing we would like to emphasize in this paper. \\
\indent
An interesting existing work in similar direction is given in~\cite{MoKh11}. As discussed in~\cite{MoKh11}, computing the achievable region with p2p codes requires exponential complexity with respect to number of users which can be a problem in a large network. The focus of~\cite{MoKh11} is, therefore, to find an efficient way of doing it by using polymatroidal nature of the achievable region, which results in polynomial time algorithm. In this paper, we restrict ourselves to two user interference networks. Note that two or three user cases provide a realistic cell-edge scenario in the standard cellular system as well as a realistic multi-points cooperation scenario in a heterogeneous network due to coordination and decoding capability issue. In the main body of our work, we carefully compute achievable rates with focus on looking at how it varies depending on channel parameters and getting insight of simple implementation in practice. \\
\indent
The rest of the paper is organized as follows. Section~\ref{sec:int} discusses the achievable region with p2p codes for standard interference channel in which traditional message assignment is used. Section~\ref{sec:gen} discusses the case when such message assignment is completely generalized. Section~\ref{sec:con} concludes the paper. 
\paragraph{Preliminaries}
Regarding the channel model, we consider general discrete memoryless channel with the form of $p(y_{1,t},y_{2,t}|x_{1,t},x_{2,t})$. Note that this includes multiple-input multiple-output (MIMO) Gaussian case. From now on, we drop subscript $t$ in channel input and output symbols for simplicity. We also define p2p codes to be standard single codebook random code ensemble. In this ensemble, each element of a codeword is generated randomly according to $p(x_i)$.
\section{Interference channel}
\label{sec:int}
We first consider the traditional interference channel in which only $(1,0),(0,2),(1,2)$ are allowed as message assignment. This case corresponds to the current cellular system. In this paper, we are interested in the maximum sum rate characterization which corresponds to characterization of some of boundary points in the achievable region. This reflects practical interest of system throughput maximization. \\
\indent
First consider how the capacity region with p2p codes is characterized for $(1,2)$. When $(1,2)$ is used, two receivers are active at the same time. Each receiver has an option of decoding interference or not. As discussed in~\cite{BaGaKi12}, let us define $\mathcal{R}_{i,IAN}$ and $\mathcal{R}_{i,SD}$ where `IAN' stands for (treat) `interference as noise' and `SD' stands for `simultaneous decoding'. Let $\mathcal{R}_{i,IAN}$ be the set of rate pairs $(R_1,R_2)$ such that
\begin{equation}
R_i \leq I(X_i;Y_i).
\end{equation}
Let $\mathcal{R}_{i,SD}$ be the set of rate pairs $(R_1,R_2)$ such that
\begin{subequations}
\begin{eqnarray}
R_i &&\leq I(X_i;Y_i|X_j),\\
R_j &&\leq I(X_j;Y_i|X_i),\\
R_i+R_j &&\leq I(X_i, X_j; Y_i),
\end{eqnarray}
\end{subequations}
where $j \neq i$. Then, the capacity region $\mathcal{R}$ in a MAC form is given as 
\begin{equation}
\label{eq:mac}
\mathcal{R}= (\mathcal{R}_{1,IAN} \cap \mathcal{R}_{2,IAN}) \cup (\mathcal{R}_{1,SD} \cap \mathcal{R}_{2,IAN}) \cup (\mathcal{R}_{1,IAN} \cap \mathcal{R}_{2,SD})
              \cup (\mathcal{R}_{1,SD} \cap \mathcal{R}_{2,SD}).
\end{equation} 
By using this, the maximum sum rates of $(1,0),(0,2),(1,2)$ can be characterized as in Table~\ref{tab:sum}.
\begin{table}[ht]
\caption{The maximum sum rate characterization in interference channel}
\begin{center}
{
\begin{tabular}{|c|c|c|c|}
  \hline
  Scheme    & $(1,0)$ & $(0,2)$ & $(1,2)$  \\
  \hline
  \hline
  Sum rate & $I(X_1;Y_1|X_2)$ & $I(X_2;Y_2|X_1)$ & $\max_{(R_1,R_2) \in \mathcal{R}} \{R_1+R_2\}$ \\
\hline
   Sum rate (SISO Gaussian) & $\log_2 \Big(1+|h_{11}|^2P\Big)$ & $\log_2 \Big(1+|h_{22}|^2P\Big)$ & See~\eqref{eq:12g}\\
\hline
\end{tabular}
}
\end{center}
\label{tab:sum}
\end{table}
In Gaussian case as well as in the actual cellular system, there is a notion of transmit signal power. Note that combination of $(0,2)$ and $(1,0)$ corresponds to TDM. When TDM is used, there is potential to use twice of power at each transmission since each transmitter remains idle for half of time. Since it may not be easy to realize such configuration in practice due to peak power constraint, we do not allow it in this paper. The case when TDM can use twice of power at each transmission is discussed in~\cite{BaLeKa11} for SISO Gaussian case. Since sum rates of $(1,0), (0,2)$ are trivial to obtain, we only look into $(1,2)$ in detail.
\subsection{Sum rates of $(1,2)$}
The capacity region of $(1,2)$ described in~\eqref{eq:mac} implies four different operational combinations depending on what each receiver does. By considering this, we can rewrite the maximum sum rate of $(1,2)$ in a min-max form as follows. When the both receivers decodes interference, sum rate is given as
\begin{align}
R_{sum}=\min \Big\{I(X_1,X_2;Y_1), I(X_1,X_2;Y_2),I(X_1;Y_1|X_2)+I(X_2;Y_2|X_1)\Big\}.
\end{align}
When receiver 1 decodes interference while receiver 2 does not decode interference or its interference decoding attempt fails (equivalence of these two is demonstrated in~\cite{BaGaKi12}), sum rate is given as
\begin{align}
R_{sum}&= I(X_1;Y_1|X_2)+\min \Big\{I(X_2;Y_1), I(X_2;Y_2) \Big\}.
\end{align}
When receiver 2 decodes interference while receiver 1 does not decode interference or its interference decoding attempt fails, sum rate is given as
\begin{align}
R_{sum}&= I(X_2;Y_2|X_1)+\min \Big\{I(X_1;Y_1), I(X_1;Y_2) \Big\}.
\end{align}
When the both receivers do not decode interferences or their interference decoding attempts fail, sum rate is given as
\begin{align}
R_{sum}= I(X_1;Y_1)+ I(X_2;Y_2).
\end{align}
Therefore, maximum sum rate of $(1,2)$ is given as
\begin{align}
\label{eq:12}
&R_{sum}\nonumber\\
&=\max \bigg\{ \min \Big\{I(X_1,X_2;Y_1), I(X_1,X_2;Y_2),I(X_1;Y_1|X_2)+I(X_2;Y_2|X_1)\Big\},\nonumber\\
&\qquad \qquad  I(X_1;Y_1|X_2)+\min \Big\{I(X_2;Y_1), I(X_2;Y_2) \Big\},\nonumber\\
&\qquad \qquad   I(X_2;Y_2|X_1)+\min \Big\{I(X_1;Y_1), I(X_1;Y_2) \Big\},\nonumber\\
&\qquad \qquad  I(X_1;Y_1)+ I(X_2;Y_2) \bigg\}.
\end{align}
In SISO Gaussian case, it becomes
\begin{align}
\label{eq:12g}
&R_{sum}\nonumber\\
&=\max \bigg\{ \min \Big\{\log_2 \Big(1+P\sum_{i=1,2} |h_{1i}|^2\Big), \log_2 \Big(1+P\sum_{i=1,2} |h_{2i}|^2\Big), \sum_{i=1,2} \log_2 \Big(1+|h_{ii}|^2P\Big)\Big\},\nonumber\\
&\qquad \qquad  \log_2 \Big(1+|h_{11}|^2P\Big)+\min \Big\{\log_2\Big(1+ \frac{|h_{12}|^2P}{1+|h_{11}|^2P}\Big), \log_2\Big(1+ \frac{|h_{22}|^2P}{1+|h_{21}|^2P}\Big) \Big\},\nonumber\\
&\qquad \qquad   \log_2 \Big(1+|h_{22}|^2P\Big)+\min \Big\{\log_2\Big(1+ \frac{|h_{21}|^2P}{1+|h_{22}|^2P}\Big), \log_2\Big(1+ \frac{|h_{11}|^2P}{1+|h_{12}|^2P}\Big) \Big\},\nonumber\\
&\qquad \qquad  \sum_{i=1,2,j\neq i} \log_2\Big(1+ \frac{|h_{ii}|^2P}{1+|h_{ij}|^2P}    \Big) \bigg\}.
\end{align}
From the second and the third term of maximization in~\eqref{eq:12} and from Table~\ref{tab:sum}, it can be easily seen that $(1,2)$ has strictly larger sum rate than $(1,0)$ or $(0,2)$, and hence, the maximum sum rate for interference channel is always given by $(1,2)$.
\subsection{Further evaluation of sum rates of $(1,2)$}
In this section, we will focus on identifying maximizer of~\eqref{eq:12} to characterize the maximum sum rate for interference channel. As mentioned earlier, the capacity region described in~\eqref{eq:mac} and the maximum sum rate expression described in~\eqref{eq:12sum} implies four different operational combinations. By acknowledging that, we can get detailed evaluation of~\eqref{eq:12sum} in a very simple manner. Let us define another tuple $(b_1,b_2)$ in which $b_1$ and $b_2$ take $\circ$ or $\times$. $b_1$ indicates whether receiver 1 decodes interference or not while $b_2$ does the same for receiver 2. For example, $(\circ, \times)$ corresponds to receiver 1 decoding interference and receiver 2 not decoding. For ease of exposition, we would also like to define several rate terms. Let $R_{ij,IF}=I(X_j;Y_i|X_{\overline{j}})$ and $R_{ij,IAN}=I(X_j;Y_i)$, where $\overline{j}=2$ if $j=1$ and $\overline{j}=1$ otherwise. $R_{ij,IF}$ corresponds to the maximum decodable rate of $x_j$ at receiver $i$ when the other message is known where IF stands for `interference free'. Similarly, $R_{ij,IAN}$ corresponds to the maximum decodable rate of $x_j$ at receiver $i$ when the other message is treated as noise where IAN stands for `interference as noise'.
\begin{theorem}
\label{thm:12}
The maximum sum rate of $(1,2)$ is given as follows.
\begin{align}
\label{eq:12sum}
R_{sum}=\begin{cases}
\sum_{i=1,2} R_{ii,IF} & \text{if } R_{12,IAN}>R_{22,IF}, R_{21,IAN}>R_{11,IF} \text{ (a) } \\
\sum_{i=1,2} R_{ii,IAN} & \text{if } R_{12,IF}<R_{22,IAN}, R_{21,IF}<R_{11,IAN} \text{ (b) } \\
R_{11,IF}+\min\{ R_{12,IAN},  R_{22,IAN}\} & \text{if } I(X_1,X_2;Y_1)>I(X_1,X_2;Y_2), R_{11,IF}>R_{21,IF},\\
                                           & \text{not if (a), (b)}\\
R_{22,IF}+\min\{ R_{11,IAN},  R_{21,IAN}\} & \text{if } I(X_1,X_2;Y_1)<I(X_1,X_2;Y_2), R_{12,IF}<R_{22,IF}, \\
                                           & \text{not if (a), (b)}\\
\min \Big\{I(X_1,X_2;Y_1), I(X_1,X_2;Y_2)\Big\} & \text{otherwise.}
\end{cases}
\end{align}
\end{theorem}
Before we prove the above result, let us try to interpret what it implies. First of all, if interfering links at the both transmitters are very strong, then interferences are decoded at the both receivers, and the system achieves interference-free sum rate. If interfering links at the both transmitters are very weak, then no receiver decodes interference. When two MAC channels at receivers do not have the same quality, only the better MAC receiver decodes interference if interference decodability at the worse MAC receiver is below certain threshold. Otherwise, the both receivers decode interference.
\begin{proof}
Consider first identifying minimizer in sum rate expression of $(\circ,\circ)$. We can easily get the following. The maximum sum rate of $(\circ,\circ)$ is
\begin{align}
R_{sum}=\begin{cases}\sum_{i=1,2} R_{ii,IF} &\text{if } R_{12,IAN}>R_{22,IF}, R_{21,IAN}>R_{11,IF}\\
\min \Big\{I(X_1,X_2;Y_1), I(X_1,X_2;Y_2)\Big\} &\text{otherwise.}
\end{cases}
\end{align}
It can easily be seen that sum rate of $(1,2)$ can never be greater than $\sum_{i=1,2} R_{ii,IF}$. Hence, we only need to consider the case when $\sum_{i=1,2} R_{ii,IF}$ is not the maximum sum rate of $(\circ,\circ)$ for further evaluation of sum rate of $(1,2)$.\\
\indent
Consider now identifying maximizer between $(\circ,\times)$ and $(\times, \times)$. In this case, the constraint on $R_1$ solely comes from decodability at receiver 1, and this enables fairly simple identification. Figure~\ref{fig:ox} describes determination process of the sum rate maximizer between $(\circ,\times)$ and $(\times, \times)$. It is based on well-known water (power)-level diagram  in SISO Gaussian case resulted from superposition of signals at each receiver. 
\begin{figure}[t]
  \includegraphics[width=\textwidth]{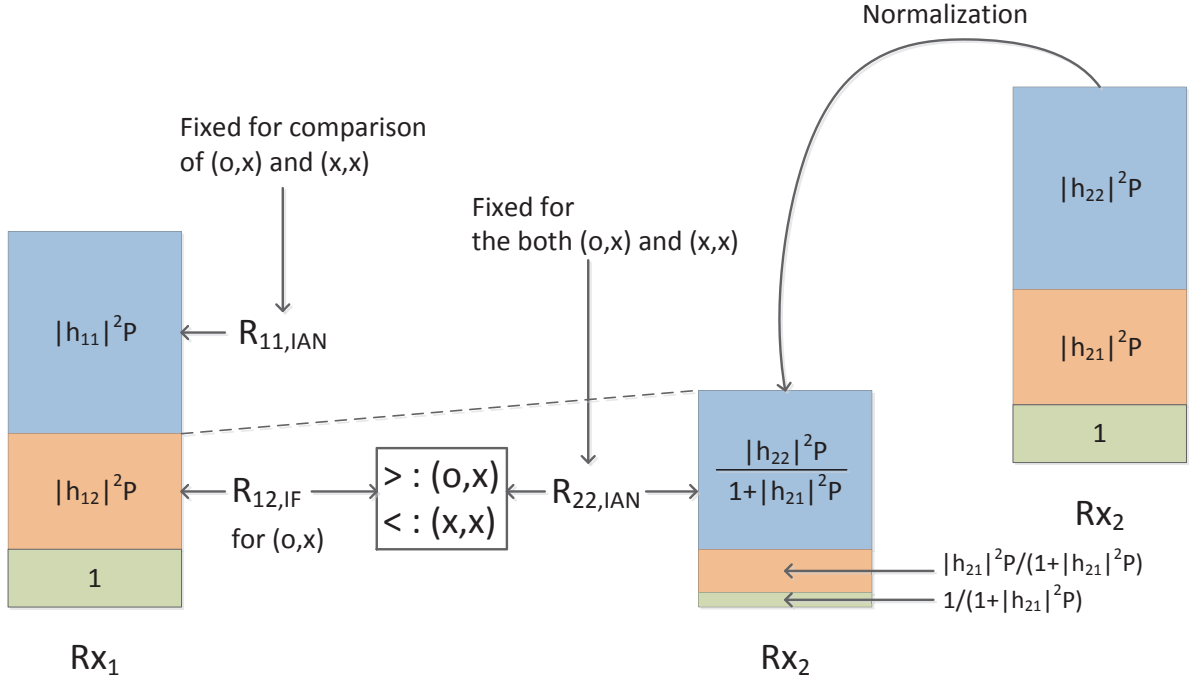}
  \caption{Determination of the sum rate maximizer between $(\circ,\times)$ and $(\times, \times)$.}
  \label{fig:ox}
\end{figure}
Although other general cases may not be expressed with such diagram, the exact same intuition carries over. Note that all we need to check here is whether decoding requirement of $x_2$ at receiver 1 prevents from setting $R_2=R_{22,IAN}$ or not. If it does, then receiver 1 should not decode interference to maximize the sum rate. This can be done by simply comparing $R_{12,IF}$ with $R_{22,IAN}$. We can do similarly for $(\times,\circ)$ and $(\times, \times)$ by comparing $R_{21,IF}$ with $R_{11,IAN}$. Consequently, $(\times, \times)$ attains the maximum sum rate among $(\circ,\times),(\times,\circ)$, $(\times, \times)$ if and only if $R_{12,IF}<R_{22,IAN}$ and $R_{21,IF}<R_{11,IAN}$. When it does, it satisfies $\sum_{i=1,2} R_{ii,IAN}>\max\{I(X_1,X_2;Y_1), I(X_1,X_2;Y_2)\}$, which means $(\times, \times)$ gives the maximum sum rate of $(1,2)$ if $R_{12,IF}<R_{22,IAN}$ and $R_{21,IF}<R_{11,IAN}$ and not if $R_{12,IAN}>R_{22,IF}$ and $R_{21,IAN}>R_{11,IF}$. We can easily see, however, that ($R_{12,IF}<R_{22,IAN}$, $R_{21,IF}<R_{11,IAN}$) and ($R_{12,IAN}>R_{22,IF}$, $R_{21,IAN}>R_{11,IF}$) are mutually exclusive.\\
\indent
It remains to characterize the maximum sum rate among $(\circ,\times),(\times,\circ),(\circ, \circ)$ when none of ($R_{12,IF}<R_{22,IAN}$, $R_{21,IF}<R_{11,IAN}$) and ($R_{12,IAN}>R_{22,IF}$, $R_{21,IAN}>R_{11,IF}$) holds. It can be seen from~\eqref{eq:12} that sum rate of $(\circ,\times)$ can never be greater than $I(X_1,X_2;Y_1)$, and sum rate of $(\times, \circ)$ can never be greater than $I(X_1,X_2;Y_2)$. This means that $(\circ,\times)$ can attain the maximum sum rate of $(1,2)$ only if $I(X_1,X_2;Y_1)>I(X_1,X_2;Y_2)$, and $(\times,\circ)$ does only if $I(X_1,X_2;Y_1)<I(X_1,X_2;Y_2)$.
We now first compare $(\circ,\times)$ and $(\circ, \circ)$ when $I(X_1,X_2;Y_1)>I(X_1,X_2;Y_2)$ holds. This is very similar to comparison of $(\circ,\times)$ and $(\times, \times)$, and Figure~\ref{fig:ox1} describes determination process of the sum rate maximizer between $(\circ,\times)$ and $(\circ, \circ)$. 
\begin{figure}[t]
  \includegraphics[width=0.8\textwidth]{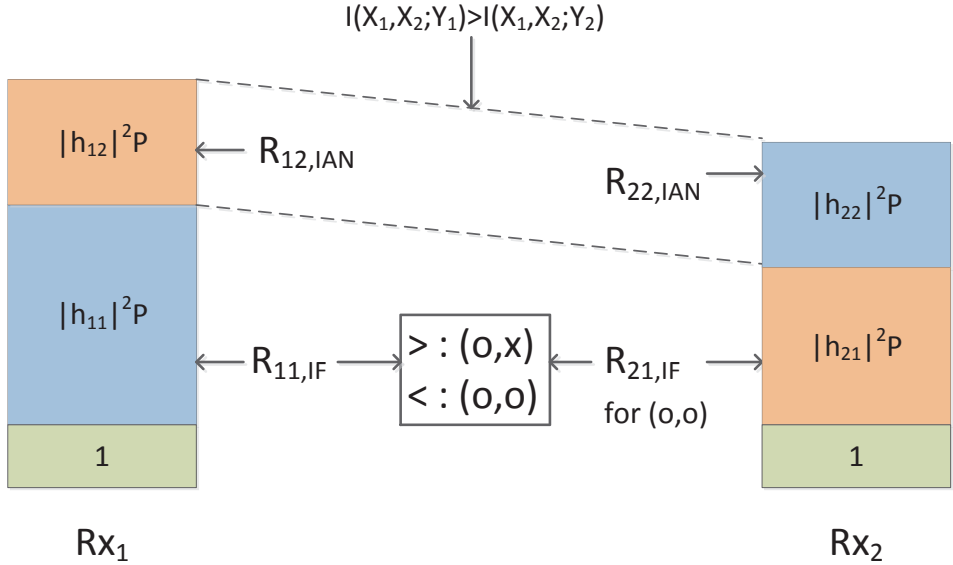}
  \caption{Determination of the sum rate maximizer between $(\circ,\times)$ and $(\circ, \circ)$.}
  \label{fig:ox1}
\end{figure}
We need to check whether decoding requirement of $x_1$ at receiver 2 becomes bottleneck of determining $R_1$ or not. If it does, then receiver 2 should not decode interference to maximize the sum rate. This can be done by simply comparing $R_{11,IF}$ with $R_{21,IF}$. We can do similarly for $(\times,\circ)$ and $(\circ, \circ)$ when $I(X_1,X_2;Y_1)<I(X_1,X_2;Y_2)$ holds by comparing $R_{12,IF}$ with $R_{22,IF}$. Consequently, we get~\eqref{eq:12sum}.
\end{proof}
Proof of Theorem~\ref{thm:12} gives an interesting implementation idea in the cellular system when only one of two receivers has capability of decoding interference. Rate determination of this case is essentially described in Figure~\ref{fig:ox}, and its implementation scenario is described in Figure~\ref{fig:imp}. 
\begin{figure}[t]
  \includegraphics[width=\textwidth]{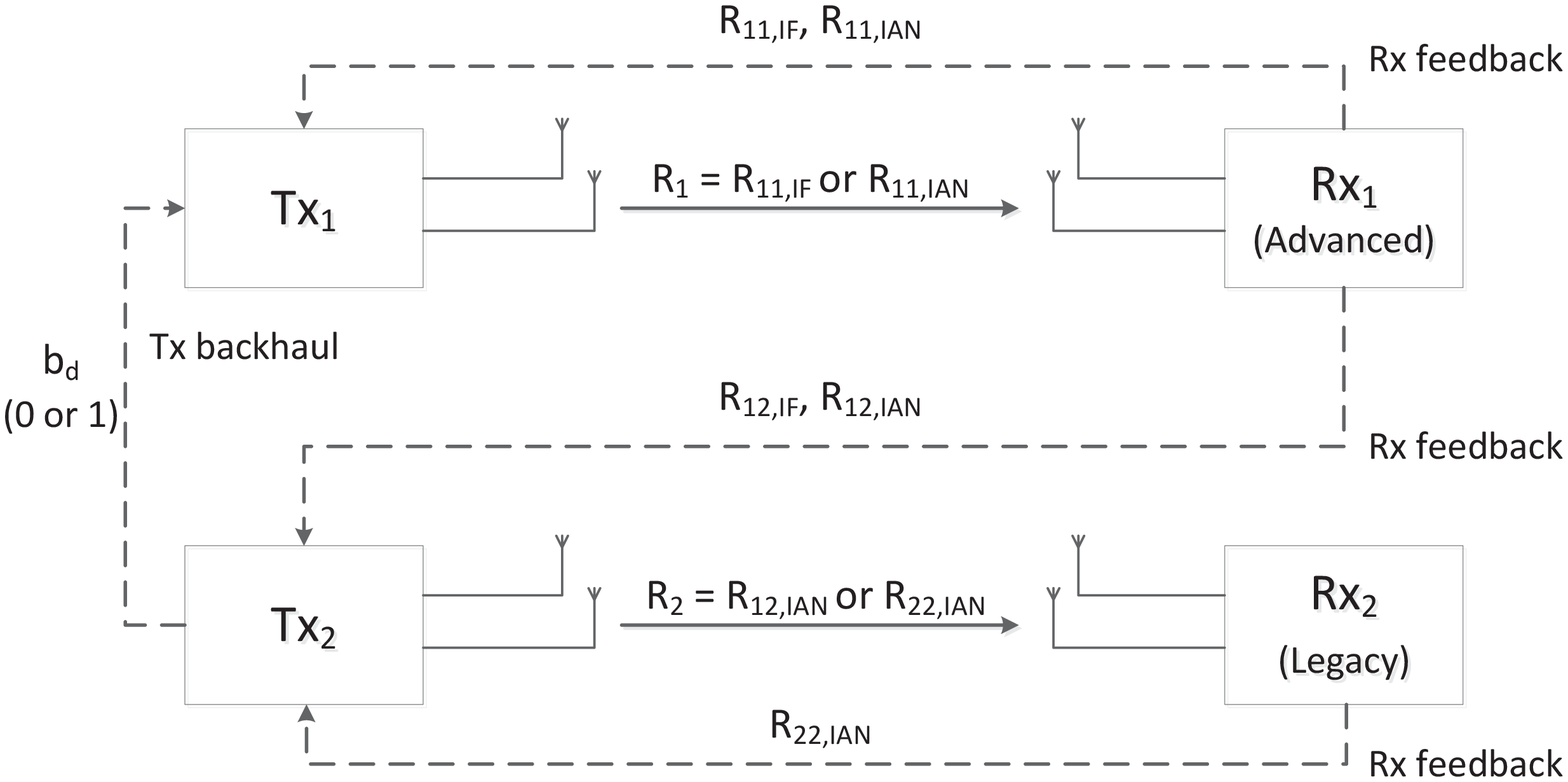}
  \caption{Implementation scenario of rate determination with one legacy and one advanced receiver: interference channel}
  \label{fig:imp}
\end{figure}
The procedure is also described in Algorithm~\ref{alg:imp}.
\begin{algorithm}
\caption{Determination of transmission rate with one legacy and one advanced receiver: interference channel}
\label{alg:imp}
\begin{algorithmic}
\Require Rx's broadcast interference decoding capability
\Require Advanced Rx (Rx 1) feedback: $R_{11,IF}$, $R_{11,IAN}$ to Tx 1
\Require Advanced Rx (Rx 1) feedback: $R_{12,IF}$, $R_{12,IAN}$ to Tx 2
\Require Legacy Rx (Rx 2) feedback: $R_{22,IAN}$ to Tx 2
\State At Tx 2,
    \If{$R_{22,IAN}>R_{12,IF}$}
        \State $b_d \Leftarrow 0$
        \State $R_2 \Leftarrow R_{22,IAN}$
    \Else [$R_{22,IAN} \leq R_{12,IF}$]
        \State $b_d \Leftarrow 1$
        \State $R_2 \Leftarrow \min\{R_{12,IAN}, R_{22,IAN} \}$
    \EndIf
\Require Tx 2 notify $b_d$ to Tx 1
\State At Tx 1,
        \If{$b_d=0$}
            \State $R_1 \Leftarrow R_{11,IAN}$
            
        \Else[$b_d=1$]
            \State $R_1 \Leftarrow R_{11,IF}$
            
        \EndIf
\end{algorithmic}
\end{algorithm}
 A fascinating nature of this procedure is that two transmitters' rate determination is almost isolated except for notification bit $b_d$ which greatly reduces backhaul overhead. \\
\indent
Unfortunately, more sophisticated joint rate determination is required when the both receivers have interference decoding capability. This procedure is essentially captured in~\eqref{eq:12sum}. In this case, the both receivers still broadcast their capabilities and feed back achievable rates as in the previous example. This means that the receiver procedure would be exactly the same no matter what type of receiver is at the other side, and only the transmitter procedure is changed which is desirable in the cellular system.
\section{Generalized interference channel}
\label{sec:gen}
In this section, we now allow more generalized message assignment at transmitters for interference channel. In other words, $(0,1),(1,1),(2,0),(2,1),(2,2)$ are now allowed tuples. Maximum sum rates of these are characterized in Table~\ref{tab:sum1} except for $(2,1)$ which can be obtained in similar ways to $(1,2)$.
\begin{table}[ht]
\caption{Maximum sum rate characterization in interference channel}
\begin{center}
{
\begin{tabular}{|c|c|c|}
  \hline
  Scheme    & $(2,0)$ & $(0,1)$ \\
  \hline
  Sum rate & $I(X_1;Y_2|X_2)$ & $I(X_2;Y_1|X_1)$ \\
\hline
   Sum rate (SISO Gaussian) & $\log_2 \Big(1+|h_{21}|^2P\Big)$ & $\log_2 \Big(1+|h_{12}|^2P\Big)$ \\
\hline
\hline
  Scheme    & $(1,1)$  & $(2,2)$\\
  \hline
  Sum rate & $I(X_1,X_2;Y_1)$ & $I(X_1,X_2;Y_2)$\\
\hline
   Sum rate (SISO Gaussian) & $\log_2 \Big(1+P(|h_{11}|^2+|h_{12}|^2)\Big)$ & $\log_2 \Big(1+P(|h_{21}|^2+|h_{22}|^2)\Big)$\\
\hline
\end{tabular}
}
\end{center}
\label{tab:sum1}
\end{table}
It can be clearly seen from Table~\ref{tab:sum1} that $(2,0)$ or $(0,1)$ has strictly smaller sum rate than those of $(1,1)$ or $(2,2)$. Furthermore, sum rates of $(\circ,\times),(\times,\circ),(\circ, \circ)$ of $(1,2)$ cannot be larger than those of $(1,1)$ or $(2,2)$ as can be seen by Table~\ref{tab:sum} and Table~\ref{tab:sum1}. Similar arguments can be made with $(2,1)$. In other words, the maximum sum rate of generalized interference channel is given by $(1,1)$, $(2,2)$, $(\times,\times)$ of $(1,2)$, $(\times,\times)$ of $(2,1)$. The following theorem characterizes the maximum sum rate of the generalized interference channel.
\begin{theorem}
\label{thm:11}
The maximum sum rate of the generalized interference channel is given as follows.
\begin{align}
\label{eq:11}
R_{sum}=\begin{cases}
\sum_{i=1,2} R_{ii,IAN} & \text{if } R_{12,IF}<R_{22,IAN}, R_{21,IF}<R_{11,IAN}\\
\sum_{i=1,2} R_{i\overline{i},IAN} & \text{if } R_{11,IF}<R_{21,IAN}, R_{22,IF}<R_{12,IAN}\\
\max \Big\{I(X_1,X_2;Y_1), I(X_1,X_2;Y_2)\Big\} & \text{otherwise.}
\end{cases}
\end{align}
\end{theorem}
\begin{proof}
We first compare $(1,1)$ and $(\times,\times)$ of $(1,2)$. It is equivalent to comparison of $(\circ,\times)$ and $(\times,\times)$ of $(1,2)$ which is described in Figure~\ref{fig:ox}. We can make similar arguments with comparison of $(2,2)$ and $(\times,\times)$ of $(1,2)$, i.e., it is equivalent to comparison of $(\times,\circ)$ and $(\times,\times)$ of $(1,2)$. Very similar arguments hold for $(2,1)$, all of which eventually lead to~\eqref{eq:11}.
\end{proof}
We can think of an algorithm (Algorithm~\ref{alg:imp1}) which is similar to Algorithm~\ref{alg:imp} for generalized interference channel when only one receiver has interference decoding capability. 
\begin{algorithm}
\caption{Determination of transmission rate with one legacy and one advanced receiver:generalized interference channel}
\label{alg:imp1}
\begin{algorithmic}
\Require Rx's broadcast interference decoding capability
\Require Advanced Rx (Rx 1) feedback: $R_{11,IF}$, $R_{11,IAN}$ to Tx 1
\Require Advanced Rx (Rx 1) feedback: $R_{12,IF}$, $R_{12,IAN}$ to Tx 2
\Require Legacy Rx (Rx 2) feedback: $R_{21,IAN}$ to Tx 1
\Require Legacy Rx (Rx 2) feedback: $R_{22,IAN}$ to Tx 2
\State At Tx 2,
    \If{$R_{22,IAN}>R_{12,IF}$}
        \State $b_{d} \Leftarrow 0$
        \State $R_2 \Leftarrow R_{22,IAN}$
    \Else [$R_{22,IAN} \leq R_{12,IF}$]
        \State $b_d \Leftarrow 1$
        \State $R_2 \Leftarrow R_{12,IAN}$
    \EndIf
\Require Tx 2 notify $b_d$ to Tx 1
\State At Tx 1,
        \If{$b_d=0$}
            \State $R_1 \Leftarrow R_{11,IAN}$
        \Else[$b_d=1$]
          \If{$R_{21,IAN}>R_{11,IF}$}
           \State $R_1 \Leftarrow R_{21,IAN}$
          \Else [$R_{21,IAN} \leq R_{11,IF}$]
           \State $R_1 \Leftarrow R_{11,IF}$
            \EndIf
        \EndIf
\end{algorithmic}
\end{algorithm}
We can also consider this when the candidate receiver of joint transmission is predetermined, i.e., only one of $(1,1)$ and $(2,2)$ is allowed.\\
\indent
The fact that joint transmission provides the maximum sum rate of generalized interference channel unless interference is very weak with the traditional message assignment demonstrates advantage of more liberal message assignments combined with advanced receivers in the cellular system. Note that even with the traditional message assignment, $(1,2)$ can also attain the same sum rate as joint transmission in some cases as can be seen in~\eqref{eq:12sum}. First, when $I(X_1,X_2;Y_1)=I(X_1,X_2;Y_2)$, i.e., two MAC channels at receivers are of the same quality which would be prevalent at the exact cell-edge in the cellular system. Note that the both receivers should have interference decoding capability in this case. Secondly, when $I(X_1,X_2;Y_1)\neq I(X_1,X_2;Y_2)$. Without loss of generality, let us assume $I(X_1,X_2;Y_1)> I(X_1,X_2;Y_2)$. In this case, $(1,2)$ attains the same maximum rate as joint transmission if interference is not too strong ($R_{11,IF}>R_{21,IF}$ and $R_{12,IAN}<  R_{22,IAN}$) when at least one receiver has interference decoding capability.
\section{Conclusion}
\label{sec:con}
As seen in this paper, the capacity region for interference networks with p2p codes can be evaluated in simple ways by acknowledging operational characteristics in the expression of such region. This type of detailed evaluation provides considerably more insights on channel condition which activates specific decoding operations than what MAC or min-max form does. Such insights can be used to set up a meaningful guideline for implementation of interference coordination in the wireless system while maintaining the current encoding and decoding structure. \\
\indent
As mentioned earlier, a three user case would also be interesting to consider although evaluation of such case would be considerably more complicated than that of two user case. Due to detection and decoding complexity and performance limitation in the practical MIMO system, it is unlikely that a joint decoding receiver with capability of decoding all incoming layers from three transmitters is practically realized. In this case, restricting decodable user set to be two would be viable alternative, which would also result in much simpler evaluation. 

\bibliographystyle{IEEEtran}
\bibliography{bae}

\begin{thebibliography}{1}
\providecommand{\url}[1]{#1}
\csname url@samestyle\endcsname
\providecommand{\newblock}{\relax}
\providecommand{\bibinfo}[2]{#2}
\providecommand{\BIBentrySTDinterwordspacing}{\spaceskip=0pt\relax}
\providecommand{\BIBentryALTinterwordstretchfactor}{4}
\providecommand{\BIBentryALTinterwordspacing}{\spaceskip=\fontdimen2\font plus
\BIBentryALTinterwordstretchfactor\fontdimen3\font minus
  \fontdimen4\font\relax}
\providecommand{\BIBforeignlanguage}[2]{{%
\expandafter\ifx\csname l@#1\endcsname\relax
\typeout{** WARNING: IEEEtran.bst: No hyphenation pattern has been}%
\typeout{** loaded for the language `#1'. Using the pattern for}%
\typeout{** the default language instead.}%
\else
\language=\csname l@#1\endcsname
\fi
#2}}
\providecommand{\BIBdecl}{\relax}
\BIBdecl

\bibitem{BaGaTs11}
F.~Baccelli, A.~{El Gamal}, and D.~N.~C. Tse, ``Interference networks with
  point-to-point codes,'' \emph{{IEEE} Trans. Inf. Theory}, vol.~57, pp.
  2582--2596, May 2011.

\bibitem{BaGaKi12}
B.~Bandemer, A.~{El Gamal}, and Y.-H. Kim, ``Optimal achievable rates for
  interference networks with random codes,'' \emph{Arxiv preprint arXIV:
  1210.4596}, 2012.

\bibitem{HaKo81}
T.~S. Han and K.~Kobayashi, ``A new achievable rate region for the interference
  channel,'' \emph{{IEEE} Trans. Inf. Theory}, vol. IT-27, pp. 49--60, Jan.
  1981.

\bibitem{Sa04}
I.~Sason, ``On achievable rate regions for the {G}aussian interference
  channel,'' \emph{{IEEE} Trans. Inf. Theory}, vol.~50, pp. 1345--1356, Jun.
  2004.

\bibitem{ShCh07}
X.~Shang and B.~Chen, ``A new computable achievable rate region for the
  {G}aussian interference channel,'' in \emph{Proceedings of ISIT 2007}, Nice,
  France, Jun. 2007.

\bibitem{EtTsWa08}
R.~H. Etkin, D.~N.~C. Tse, and H.~Wang, ``Gaussian interference channel
  capacity to within one bit,'' \emph{{IEEE} Trans. Inf. Theory}, vol.~54, pp.
  5534--5562, Dec. 2008.

\bibitem{BaLeKa11}
J.~H. Bae, J.~Lee, and I.~Kang, ``Simple transmission strategies for
  interference channel,'' in \emph{Proceedings of ISIT 2012}, Cambridge, MA,
  USA, Jul. 2012.

\bibitem{MoKh11}
A.~S. Motahari and A.~K. Khandani, ``To decode the interference or to consider
  it as noise,'' \emph{{IEEE} Trans. Inf. Theory}, vol.~57, pp. 1274--1283,
  Mar. 2011.

\end{thebibliography}
\end{document}